\documentclass{article}

\usepackage{fullpage}
\usepackage{amsthm,amssymb,amsmath}  
\usepackage{comment,xspace,enumerate}
\usepackage[capitalise]{cleveref}
\usepackage[utf8]{inputenc}
\usepackage{thmtools}
\usepackage{thm-restate}
\usepackage{authblk}

  \theoremstyle{plain}
  \newtheorem{theorem}{Theorem}[section]
  \newtheorem{lemma}[theorem]{Lemma}  
  \newtheorem{corollary}[theorem]{Corollary}  
  \newtheorem{fact}[theorem]{Fact}
  \newtheorem{observation}[theorem]{Observation}
  \newtheorem{lem}[theorem]{Lemma}
  \renewenvironment{lemma}{\begin{lem}}{\end{lem}}
  \crefname{lem}{Lemma}{Lemmas}
  \newtheorem{cor}[theorem]{Corollary}
  
  \crefname{cor}{Corollary}{Corollaries}
  \theoremstyle{definition}
  \newtheorem{definition}[theorem]{Definition}
   \newtheorem{defi}[theorem]{Definition}
  \renewenvironment{definition}{\begin{defi}}{\end{defi}}
  \crefname{defi}{Definition}{Definitions}
  \newtheorem{example}[definition]{Example}
 
  \newtheorem*{claim}{Claim}

\title{Efficient Index for Weighted Sequences}

\author[1]{Carl Barton}
\author[2]{Tomasz Kociumaka\footnote{Supported by the Polish Ministry of Science and Higher Education under the `Iuventus Plus' program in 2015-2016 grant no 0392/IP3/2015/73.}$^{,}$}
\author[3]{Solon P. Pissis}
\author[2,3]{Jakub Radoszewski\footnote{The author is a Newton International Fellow.}$^{,*,}$}

\affil[1]{The Blizard Institute, Barts and The London School of Medicine and Dentistry,
    Queen Mary University of London, UK\\
    \texttt{carl.barton@qmul.ac.uk}}
\affil[2]{Institute of Informatics, University of Warsaw, Warsaw, Poland\\
    \texttt{[kociumaka,jrad]@mimuw.edu.pl}}
\affil[3]{Department of Informatics, King's College London, London, UK\\
    \texttt{solon.pissis@kcl.ac.uk}
}

\date{\vspace{-5ex}}

\newsavebox{\mybox}
\newenvironment{problem}[1]
{\begin{center}\begin{lrbox}{\mybox}\begin{minipage}{0.96\columnwidth}#1\\}
{\end{minipage}\end{lrbox}\fbox{\usebox{\mybox}}\end{center}}

  \usepackage{amsmath,amsfonts}
  \usepackage{tikz}
  
  \usepackage[T1]{fontenc}
  \usepackage[utf8]{inputenc}
  \usepackage{comment}
  \usepackage{xspace}
  \usepackage{enumerate}
  \usepackage{listings}
  \usepackage{multirow}
  \usepackage{todonotes}
    \usepackage{thmtools}
\usepackage{thm-restate}
  \usepackage[noend, boxed, noline]{algorithm2e}
  \usepackage[capitalise]{cleveref}

  \newcommand{\Oh}{\mathcal{O}}

  \newcommand{\mayqed}{}

  \newcommand{\Occ}{\mathit{Occ}_{\frac1z}}
  \newcommand{\fr}{$\frac1z$}
  \newcommand{\match}{\approx_{\frac1z}}
  
  \renewcommand{\P}{\mathcal{P}}

  \newcommand{\T}{\mathcal{T}}
  \renewcommand{\S}{\mathcal{S}}
  \newcommand{\X}{\mathbf{X}}
  \newcommand{\E}{\mathcal{E}}
  \newcommand{\I}{\mathcal{I}}
  \newcommand{\D}{\mathcal{D}} 

  \newcommand{\OL}{\mathit{OL}}
  \newcommand{\PREF}{\mathit{WPT}}

  \renewcommand{\root}{\textsl{root}}
  \newcommand{\str}{\textsl{str}}
  \newcommand{\lca}{\textsl{lca}}
  \newcommand{\back}{\textsl{back}}
  \newcommand{\strback}{\textsl{str-back}}
  \newcommand{\piback}{\pi\textsl{-back}}
   
  \newcommand{\eend}{\textsl{end}}
  \newcommand{\len}{\textsl{len}}
  \newcommand{\maxlen}{\textsl{maxlen}}

  \newcommand{\wlcp}{\textsl{wlcp}}
  \newcommand{\maxgap}{\textsl{maxgap}}
  \newcommand{\depth}{\textsl{depth}}
  \newcommand{\covercheck}{\textsl{cover-check}}

\begin{document}
  \maketitle
  
\begin{abstract}
  The problem of finding factors of a text string which are identical or similar to a given pattern string is a central problem in computer science.
  A generalised version of this problem consists in implementing an index over the text to support efficient on-line pattern queries.
  We study this problem in the case where the text is {\em weighted}: for every position of the text and every letter of the alphabet a probability of
  occurrence of this letter at this position is given. Sequences of this type, also called position weight matrices, are commonly used to represent imprecise or uncertain data.
  A weighted sequence may represent many different strings, each with probability of occurrence equal to the product of probabilities of its letters at subsequent positions.
  Given a probability threshold $\frac1z$, we say that a pattern string $P$ {\em matches} a weighted text at position $i$
  if the product of probabilities of the letters of $P$ at positions $i,\ldots,i+|P|-1$ in the text is at least $\frac1z$. 
  In this article, we present an $O(nz)$-time construction of an $O(nz)$-sized index that can answer pattern matching queries in a weighted text in optimal time improving upon the state of the art
  by a factor of $z \log z$. Other applications of this data structure include an $O(nz)$-time construction of the weighted prefix table and an $O(nz)$-time computation of all covers 
  of a weighted sequence, which improve upon the state of the art by the same factor. 
\end{abstract}

  \section{Introduction} \label{sec:intro}
  Finding factors of a {\em text} resembling a {\em pattern} 
  constitutes a classical problem in computer science. Apart from its theoretical interest, it is the core computation of many applications~\cite{ASM_Surv} such as search engines, bioinformatics, processing natural language and database search.  
  
  In many situations the text can be considered as fixed and the patterns may arrive later. The algorithmic challenge is then to provide fast and direct access to all the factors of the text via the implementation of an {\em index}. 
%
The most widely used data structures for this purpose are the {\em suffix tree} and the {\em suffix array}~\cite{AlgorithmsOnStrings}.  
These data structures can be constructed in $\Oh(n)$ time for a text of length $n$. Then all locations of a pattern of length
$m$ can be found in the optimal time $\Oh(m+\mathit{Occ})$, where \textit{Occ} is the number of occurrences.

The pattern matching problem for {\em uncertain} sequences has been less explored \cite{DBLP:conf/sdm/LiBKP14}.
In this work we consider a type of uncertain sequences called {\em weighted sequences} (also known as position weight matrices, PWM).
In a weighted sequence every position contains a subset of the alphabet and every letter is assigned a probability of occurrence such that at each position the probabilities
sum up to 1.
Such sequences are common in various applications: (i) data measurements, such as imprecise sensor measurements; (ii) flexible modelling of sequence, such as binding profiles of molecular sequences; (iii) observations are private and thus sequences of observations may have artificial uncertainty introduced deliberately.

In the \textit{weighted pattern matching} (WPM) problem we are given a string of length $m$ called a pattern, a weighted sequence of length $n$ called a text, both over an alphabet $\Sigma$ of size $\sigma$, and a \emph{threshold probability} \fr. The task is to find all positions in the text where the fragment of length $m$ represents the pattern with probability at least \fr. Each such position is called an \emph{occurrence} of the pattern; we also say that the fragment and the pattern \emph{match}. 
An $\Oh(\sigma n \log m)$-time solution for the WPM problem based on Fast Fourier Transform was proposed in \cite{KCL_publication}. This problem was also considered in \cite{amir_weighted_property_matching_j} where a reduction to property matching in a text of size $\Oh(n z^2 \log z)$ was proposed.
  
In this article, we are interested in the indexing version of the WPM problem,
that is, constructing an index to provide efficient procedures for answering queries related to the content of a fixed weighted sequence.
In \cite{costas_weighted_suffix_tree_j}, the authors presented the {\em weighted suffix tree} allowing $\Oh(m + \textit{Occ})$-time WPM queries; 
the construction time and size of that data structure is $\Oh(n \sigma^{z \log z})$.
A direct application of the results in \cite{amir_weighted_property_matching_j} reduces the construction time and the size of that index to $\Oh(n z^2 \log z)$. 
The index structure built in \cite{costas_weighted_suffix_tree_j} consists of a compacted trie of all of the factors with probability greater than or equal to $\frac{1}{z}$.
A similar---though more general---indexing data structure, which assumes $z = \Oh(1)$, was also presented in~\cite{DBLP:journals/corr/ThankachanPSB15} with query time $\Oh(m + m \times \mathit{Occ})$.
Here we propose a tree-like data structure that is similar to the aforementioned ones which is, however, constructed and stored much more efficiently.

  \subparagraph*{Our model of computations}
  We assume word-RAM model with word size $w = \Omega(\log(nz))$.
  We consider the log-probability model of representations of weighted sequences in which
  probabilities can be multiplied exactly in $\Oh(1)$ time.

  A common assumption in practice is that $\sigma = \Oh(1)$ since the most commonly
  studied alphabet is $\Sigma=\{\mathtt{A},\mathtt{C},\mathtt{G},\mathtt{T}\}$.
  In this case a weighted sequence of length $n$ has a representation of $\Oh(n)$ size.
  We describe the indexing data structure under this assumption.
  In the Conclusions Section we briefly discuss the construction of the index for larger alphabets.

  \subparagraph*{Our contribution}  We present an $\Oh(nz)$-time construction of an $\Oh(nz)$-sized index that answers weighted pattern matching queries
  in optimal $\Oh(m+\mathit{Occ})$ time improving upon~\cite{amir_weighted_property_matching_j} by a factor of $z \log z$.
  Applications of our data structure include an $\Oh(nz)$-time construction of the {\em weighted prefix table}
  and an $\Oh(nz)$-time computation of all \emph{covers} of a weighted sequence, which improve upon \cite{DBLP:conf/cwords/BartonP15}
  and \cite{costas_weighted_suffix_tree_j}, respectively, by the same factor in the complexity.
  
  \subparagraph*{Structure of the paper}
  In \cref{sec:prelim} basic notation related to weighted sequences, tries and compacted tries is presented.
  In particular, we introduce an important notion of extensions of solid prefixes, which is then used to construct
  an intermediate data structure that is crucial to our index, called solid factor trie, in \cref{sec:trie}.
  The weighted index is described in \cref{sec:index}.
  First, in \cref{subsec:trie}, we show how the main component of the index, the compacted trie
  of maximal solid factors, is obtained from the solid factor trie, and then, in \cref{subsec:index},
  a black-box description of the weighted index together with all the auxiliary data structures is given.
  \cref{sec:applications} contains two examples of applications of the weighted index.
  We end with a Conclusions Section where we sketch changes to be made to the index in the case
  of a superconstant-sized integer alphabet.
  \cref{app:figures} contains an example of our data structure for a sample weighted sequence.
  
  \section{Preliminaries}\label{sec:prelim}

  Let $\Sigma=\{s_1,s_2,\ldots,s_{\sigma}\}$ be an alphabet.
  A \emph{string} $X$ over $\Sigma$ is a finite sequence of letters from $\Sigma$.
  By $X[i]$, for $1 \le i \le |X|$, we denote the $i$-th letter of $X$.
  The {\em empty string} is denoted by $\varepsilon$.
  By $X[i..j]$ we denote the string $X[i] \ldots X[j]$ called a \emph{factor} of $X$ (if $i>j$, then the factor is an empty string).
  A factor is called a \emph{prefix} if $i=1$ and a \emph{suffix} if $j=|X|$.
  A factor $Y$ of a string $X$ is {\em proper} if $Y \neq X$.  
  By $X^R$ we denote the reversal (the mirror image)
  of $X$.

  \begin{definition}[Weighted sequence]\label{def:weighted_sequence}
    A weighted sequence $X=x_1x_2 \ldots x_n$ of length $|X|=n$ over an alphabet $\Sigma=\{s_1,s_2,\ldots,s_{\sigma}\}$
    is a sequence of sets of pairs of the form:
    $$x_i = \{(s_j,\ \pi^{(X)}_i(s_j))\ :\ j \in \{1,2,\ldots,\sigma\}\}.$$
    If the considered weighted sequence is unambiguous, we write $\pi_i$ instead of $\pi^{(X)}_i$.
    Here, $\pi_i(s_j)$ is the occurrence probability of the letter $s_j$ at the position $i \in \{1,\ldots,n\}$.
    These values are non-negative and sum up to 1 for a given $i$.
  \end{definition}
  
  
  The \emph{probability of matching} of a string $P$ with a weighted sequence $X$, both of length $n$, equals
  $$\P(P,X) = \prod_{i=1}^n \pi^{(X)}_i(P[i]).$$
  We say that a string $P$ \emph{matches} a weighted sequence $X$
  with probability at least \fr, denoted by $P \match X$, if $\P(P,X) \ge \frac1z$.
  By $X[i..j]$ we denote a weighted sequence called a \emph{factor} of $X$ and equal to $x_i\ldots x_j$ (if $i>j$, then the factor is an empty weighted sequence).
  We then say that a string $P$ \emph{occurs} in $X$ at position $i$ if $P$ matches the factor $X[i..i+m-1]$.
  We also say that $P$ is a \emph{solid factor} of $X$ (starting, occurring) at position $i$.
  By $\Occ(P,X)$ we denote the set of all positions where $P$ occurs in $X$.
  The main problem considered in the paper can be formulated as follows.

  \begin{problem}{\textbf{Problem 1.} Weighted Indexing}~\\
    \textbf{Input:} A weighted sequence $X$ of length $n$ over an alphabet $\Sigma$ of size $\sigma$ and a threshold probability \fr.\\
    \textbf{Queries:} For a given pattern string $P$ of length $m$, check if $\Occ(P,X) \ne \emptyset$, compute $|\Occ(P,X)|$, or report all elements of $\Occ(P,X)$.
  \end{problem}

  We say that $P$ is a \emph{(right-)maximal solid factor} of $X$ at position $i$
  if $P$ is a solid factor of $X$ at position $i$ and no string $P' = Ps$, for $s \in \Sigma$, is a solid factor of $X$ at this position.

  \begin{fact}[Amir et al.~\cite{amir_weighted_property_matching_j}]\label{fct:maximal}
    A weighted sequence has at most $z$ different maximal solid factors starting at a given position.
  \end{fact}

  For each position of a weighted sequence $X$ we define the {\em heaviest letter} as the letter with the maximum probability (breaking ties arbitrarily).
  By $\X$ we denote a string obtained from $X$ by choosing at each position the heaviest letter.
  We call $\X$ the \emph{heavy string} of $X$. 

  \subsection{Extensions of solid factors}
  Let us fix a weighted sequence $X$ of length $n$.
  If $F$ is a solid factor of $X$ starting at position $i$ and ending at position $j$, $j \ge i-1$, then
  the string $F \X[j+1..n]$ is called the \emph{extension of the solid factor} $F$.
  By $\E$ we denote the set of extensions of all solid factors of $X$.

  \begin{observation}\label{obs:alt}
    $\E$ is exactly the set of extensions of all \emph{maximal} solid factors of $X$.
  \end{observation}
  \begin{proof}
    Let $F \X[j+1..n] \in \E$ be an extension of a solid factor $F$ starting at position $i$
    and let $k \in \{j,\ldots,n\}$ be the maximum index such that $F \X[j+1..k]$ is a solid factor
    of $X$ starting at position $i$.
    Then $M = F \X[j+1..k]$ is a maximal solid factor and $F \X[j+1..n] = M \X[k+1..n]$ is its extension.
  \mayqed\end{proof}

  The following observation shows that $\E$ is closed under suffixes.

  \begin{observation}\label{obs:cut}
    If $S \in \E$, $S \ne \varepsilon$, then the longest proper
    suffix $S'$ of $S$ also belongs to $\E$.
  \end{observation}
  \begin{proof}
    Assume that $S$ is an extension of a solid factor $F$.
    If $|F| \ge 1$, then $S'$ is an extension of the longest proper suffix of $F$.
    Otherwise, $S'$ is an extension of an empty factor.
  \mayqed\end{proof}

  \subsection{Tries}
  We consider rooted labeled trees with labels on edges, called \emph{tries}.
  The labels are letters from $\Sigma$; edges going down from a single node have distinct labels.
  The root is denoted as $\root$.

  If $T$ is a trie and $u$, $v$ are its two nodes such that $v$ is an ancestor of $u$, then by
  $\str(u,v)$ we denote the string spelled by the edge labels on the path from $u$ to $v$.
  Using this notation, we say that the trie \emph{represents} the family of strings $\{\str(u,\root)\,:\, u \in T\}$.
  As usual by $\lca(x,y)$ we denote the lowest common ancestor of the nodes $x$ and $y$.
  By $L_i$ for $i \ge 0$ we denote the $i$-th \emph{level} of $T$ that consists of nodes at depth $i$
  in the trie.

  A \emph{compacted trie} is a trie in which maximal paths whose inner nodes have degree 2 are
  represented as single edges with string labels.
  Usually such labels are not stored explicitly, but as pointers to a base string (or base strings).
  The remaining nodes are called \emph{explicit} nodes, whereas the nodes that are removed due to
  compactification are called \emph{implicit} nodes.
  A well-known example of a compacted trie is a suffix tree of a string.

  A \emph{suffix tree of a trie} $T$, denoted as $\S(T)$, is a compacted trie of the strings
  $\str(u,\root)$ for $u \in T$; see \cite{DBLP:journals/tcs/Breslauer98,DBLP:conf/isaac/Shibuya99,journals/ieice/Shibuya03}.
  The explicit nodes of $\S(T)$ that correspond to $\str(u,\root)$ for $u \in T$ are
  called \emph{terminal} nodes.
  The string labels of the edges of $\S(T)$ are not stored explicitly, but correspond to upward paths
  in the trie $T$.
  For a node $v$ of $\S(T)$, by $\str(v)$ we denote the concatenation of labels of the edges
  from the root of $\S(T)$ to $v$.

  \begin{fact}[Shibuya~\cite{journals/ieice/Shibuya03}]\label{fct:suftree_of_tree}
    The suffix tree of a trie with $N$ nodes has size $\Oh(N)$ and can be constructed in $\Oh(N)$ time.
  \end{fact}

  \section{Solid factor trie}\label{sec:trie}
  For a weighted sequence $X$ of length $n$, a \emph{solid factor trie} of $X$, denoted as $\T$,
  is a trie representing the \emph{reversals of the strings from $\E$}.
  By this definition:

  \begin{observation}
    If $S$ is a solid factor of $X$, then there exist nodes $u$, $v$ in $\T$ such that $\str(u,v)=S$.
  \end{observation}

  It turns out that the solid factor trie represents all maximal solid factors of $X$ much more efficiently than if each of them
  was stored separately.

  \begin{lemma}\label{lem:size}
    The solid factor trie $\T$ has at most $z$ nodes at each level.
  \end{lemma}
  \begin{proof}
    By \cref{obs:cut}, each node at the level $i$ in $\T$ comes from a string of length $i$ in $\E$.
    By \cref{obs:alt,fct:maximal}, there are at most $z$ strings of length $i$ in $\E$.
  \mayqed\end{proof}

  We proceed with a construction of the solid factor trie in time linear in the size of the trie.
  For this, we need to equip the data structure with additional values;
  these enhancements will also turn out useful in the construction of the weighted index.

  For each edge of the trie we store, in addition to its letter label, its probability defined as
  the probability of this letter at the respective position in $X$.
  If $v$ is an ancestor of $u$, then by $\pi(u,v)$ we denote the product of probabilities
  of edges on the path from $u$ to $v$.
  Let $H$ be the heavy path in $\T$ that corresponds to $\X$ and let $h$ be the leaf on this path.
  For each node $v$ of $\T$ we retain the node $\back(v)$ defined as $\lca(v,h)$ and
  the probability $\piback(v) = \pi(v,\back(v))$.
  We also denote $\strback(v) = \str(v,\back(v))$ (those values are not stored).

  \begin{figure}[b!]
    \begin{center}
      \begin{minipage}[b]{13cm}
      \vspace{0pt}
      \begin{algorithm}[H]
      \medskip
      \KwSty{Algorithm} Construct-$\T$($X$, $n$)\\
      \medskip
      $L_0:=\{\root\}$; $h_0:=\root$\;
      \For{$i:=1$ \KwSty{to} $n$}{
        Create a new node $h_i$ being a child of $h_{i-1}$ with the letter $\X[n-i+1]$\;
        $\back(h_i):=h_i$\;
        $\piback(h_i):=1$\;
        $L_i := \{h_i\}$\;
        \ForEach{$v \in L_{i-1}$}{
          \ForEach{$s \in \Sigma$ in order of non-increasing $\pi_{n-i+1}^{(X)}(s)$}{
            \lIf{$v = h_{i-1}$ \KwSty{and} $s= \X[n-i+1]$}{\KwSty{continue}}
            \If{$\pi_{n-i+1}^{(X)}(s) \cdot \piback(v) \ge \frac1z$}{
              Create a new node $u$ being a child of $v$ with the letter $s$\;
                $\back(u) := \back(v)$\;
                $\piback(u):=\pi_{n-i+1}^{(X)}(s) \cdot \piback(v)$\;
                $L_i := L_i \cup \{u\}$\;
              }\lElse{\KwSty{break}}
              
            }
          }
        }
      \medskip
    \end{algorithm}
    \end{minipage}
    \end{center}
  \end{figure}

  \begin{theorem}\label{thm:solid_factor_trie}
    The solid factor trie $\T$ of a weighted sequence $X$ of length $n$ can be constructed in $\Oh(nz)$ time.
  \end{theorem}
  \begin{proof}
    The trie is constructed by the algorithm Construct-$\T$($X$, $n$).
    We add new nodes to $\T$ level by level.
    A node $v$ at level $i-1$ receives a child with an edge labeled by a letter $s$ if and only if $s\,\strback(v)$
    is a solid factor at position $n-i+1$; this condition is checked using the $\piback(v)$ values.
    Then we assign the child its values of $\back$ and $\piback$.

    The correctness of the algorithm follows from the claim below.

    \begin{claim}
      After the $i$-th step of the outmost loop of the algorithm Construct-$\T$($X$, $n$),
      the trie represents the reversals of the strings from $\E$ of length at most $i$.
    \end{claim}
    \begin{proof}
      The proof goes by induction on $i$.
      The case of $i=0$ is trivial.
      Let us assume that the claim holds for $i-1$ and prove that it then also holds for $i$.
      We need to show that if a node $u$ is created by the algorithm at the $i$-th level, then $\str(u,\root) \in \E$
      and, conversely, if $S \in \E$ is a string of length $i$, then a node $u$ such that $\str(u,\root) = S$ is
      created by the algorithm at the $i$-th level.
      We prove the two implications separately.

      $(\Rightarrow)$
      If the node $u$ is created for some letter $s$, then, by the inductive hypothesis and the condition checked in the algorithm,
      $s\,\strback(v)$ is a solid factor of $X$ starting at position $n-i+1$.
      Let $j$ be the level of the node $\back(v)$.
      Then:
      $$\str(u,\root) = s\,\strback(v)\,\X[n-j+1..n] \in \E.$$

      $(\Leftarrow)$
      Let $S'$ be the longest proper suffix of $S$.
      Then $S' \in \E$ due to \cref{obs:cut}.
      By the inductive hypothesis, there exists a node $v$ in $L_{i-1}$ such that $\str(v,\root) = S'$.
      Then $S$ is an extension of the solid factor $s\, \strback(v)$, so indeed
      $\pi_{n-i+1}^{(X)}(s) \cdot \piback(v) \ge \tfrac1z$
      and the node $u$ corresponding to $S$ will be created.
    \mayqed\end{proof}

    Let us proceed with the complexity analysis.
    At each step of the innermost foreach-loop either a new node is created or the execution of the loop
    is interrupted.
    For a given $i$, the former takes place $|L_i|$ times in total and the latter takes place at most $|L_{i-1}|$ times in total.
    The whole algorithm works in $\Oh(\sum_{i=0}^n |L_i|) = \Oh(nz)$ time due to \cref{lem:size}.
  \mayqed\end{proof}
  
  Before we proceed to the construction of a weighted index, we need to introduce additional values to the solid factor trie
  that enable recovering the maximal solid factors of $X$.
  For a node $u$ of $\T$, by $\eend(u)$ we denote its ancestor $v$ such that $\str(u,v)$
  is a maximal solid factor at position $n-i+1$ in $X$, where $u \in L_i$.
  Moreover, by $\len(u)$ we denote $|\str(u,v)|$.

  \begin{lemma}\label{lem:end_len}
    The values $\eend(u)$ and $\len(u)$ for all nodes $u$ of $\T$ can be computed in $\Oh(nz)$ time.
  \end{lemma}
  \begin{proof}
    Clearly, it suffices to focus on the $\eend$-pointers, as the $\len$-values can be computed from these pointers
    in linear time if only we store for each node its level in the trie.

    For each node $u$, $\eend(u)$ is an ancestor of $\back(u)$ (possibly equal to $\back(u)$),
    therefore it is located on the heavy path $H$.
    For each node $v \in H$ from the leaf $h$ up to the root we will set the $\eend$-pointers
    for all nodes $u$ such that $\eend(u)=v$.
    In the computation we use the following property of the pointers:

    \begin{observation}\label{obs:order}
      If $x$ is an ancestor of $y$, then $\eend(x)$ is an ancestor of $\eend(y)$.
    \end{observation}

    A node will be called \emph{active} if it is a descendant of $v$ such that its $\eend$-pointer has not been computed yet but its
    children's $\eend$-pointers have all been computed.
    After a node $v \in H$ has been considered, a set $A$ containing all the active nodes $u$ together
    with the values $\pi(u,v)$ is stored.
    Initially the set is empty.

    For the next node $v \in H$ we first update the set $A$.
    If $v=h$, then we simply insert $v$ to $A$ with the probability 1.
    Otherwise, we iterate through all the nodes $u$ in the set $A$ and
    multiply their probabilities by the probability of the edge $\pi(v',v)$ where $v'$ is
    the child of $v$ on the heavy path.
    Then we insert to $A$ all the leaves in the subtrees of $\T$ corresponding to children
    of $v$ other than $v'$; their probabilities in $A$ are the values of $\piback$.

    Next, we try to set the $\eend$-pointers to the elements of $A$ and their ancestors.
    If $v$ is the root, we simply set the pointers to the root to all the elements of $A$ and
    their ancestors.
    Otherwise, let $w \in H$ be the parent of $v$.
    We iterate through all the elements $u \in A$ and for each of them check if $\pi(u,w) = \pi(u,v) \pi(v,w)$ is at least \fr.
    If so, we simply leave $u$ in $A$ for the next iterations.
    Otherwise, we set $\eend(u)=v$. 
    If $u$ was the last child of its parent for which we computed the $\eend$-pointer,
    we add the parent of $u$ to $A$.
    In order to efficiently check this condition, each node counts its children whose $\eend$-pointer is yet to be determined. 

    The correctness of the algorithm follows from \cref{obs:order}.
    The running time is proportional to the total number of times a node from $A$ is visited.
    When a node $v \in H$ is considered, for each node $u \in A$ either its $\eend$-pointer is set, which
    obviously happens at most $|\T| = \Oh(nz)$ times in total, or $\str(u,v)$ corresponds to a left-maximal solid factor
    ending at position corresponding to the level of $v$ in $\T$, which can happen at most $z$
    times by \cref{fct:maximal}.
    This implies $\Oh(nz)$ time complexity of the whole algorithm.
  \mayqed\end{proof}

  \section{Construction of the Weighted Index}\label{sec:index}
  Our index for a weighted sequence $X$ is based on the compacted trie of all maximal solid factors of $X$.
  We first show how this compacted trie can be constructed from the suffix tree $\S(\T)$ of the solid factor trie $\T$.
  Next, we describe in detail all the components of the resulting weighted index.

  \subsection{Compacted trie of maximal solid factors}\label{subsec:trie}
  First of all, from \cref{fct:suftree_of_tree,lem:size} we obtain an efficient construction of $\S(\T)$:

  \begin{lemma}\label{lem:suffix_tree_of_T}
    The suffix tree of the solid factor trie can be constructed in $\Oh(nz)$ time.
  \end{lemma}
  We assume that each terminal node of $\S(\T)$ stores, as its label, the starting position in $X$
  of the string from $\E$ that it represents.

  The trie $\T$ represents more than the (maximal) solid factors of $X$, and so does $\S(\T)$.
  However, the $\len$-values that we computed in $\T$ let us delimit the maximal solid factors.
  Using them we can transform $\S(\T)$ into a compacted trie $\T'$ of all maximal solid factors of $X$.
  In this compacted trie a terminal's label is a list of starting positions in $X$ of occurrences of the corresponding
  maximal solid factor.

  \begin{theorem}
    A compacted trie $\T'$ of all maximal solid factors of a weighted sequence $X$ of length $n$ can be constructed in $\Oh(nz)$ time.
  \end{theorem}
  \begin{proof}
    We start by constructing the solid factor trie $\T$ of $X$, together with the $\len$-values, and
    its suffix tree $\S(\T)$.
    By \cref{thm:solid_factor_trie,lem:end_len,lem:suffix_tree_of_T}, these steps take $\Oh(nz)$ time.
    Now it suffices to properly trim $\S(\T)$.
    For a terminal node $v$ in $\S(\T)$ corresponding to $\str(u,\root)$ in $\T$, as $\len(v)$ we store $\len(u)$.
    Then we need to ``lift'' such a terminal node to depth $\len(v)$ in $\S(\T)$.
    In practice we proceed it as follows.
    
    For an (explicit or implicit) node $u$ of $\S(\T)$, by $\maxlen(u)$
    we denote the maximum value of $\len(v)$ for a descendant terminal node $v$.
    As a result of trimming we leave only those (explicit or implicit) nodes $u$ for which
    $\maxlen(u)$ is at least as big as their depth in the trie; we call such nodes \emph{relevant} nodes
    and the remaining nodes \emph{irrelevant} nodes.

    This procedure can be implemented in linear time.
    Indeed, the $\maxlen$-values for all explicit nodes can be computed with a single bottom-up traversal.
    In another bottom-up traversal, we consider all irrelevant explicit nodes.
    Let $w$ be an irrelevant node and let $v$ be its parent.
    Assume that $v$ is located at depth $d$.
    If $\maxlen(w) \le d$, $w$ is removed from $\S(\T)$ and its label is appended to its parent's label.
    Otherwise, we cut the edge connecting $v$ and $w$ at depth $\maxlen(w)$ and move the irrelevant node $w$ there,
    making it relevant.
  \mayqed\end{proof}

  \subsection{The Weighted Index}\label{subsec:index}
  As already mentioned, our weighted index is based on the compacted trie $\T'$ of all maximal solid factors of $X$.
  We also need to store the solid factor trie $\T$ which lets us access the string labels of the edges
  of the compacted trie.
  For convenience we extend each maximal solid factor in $\T'$ by a symbol $ \$\not\in \Sigma$.
  As a result, each maximal solid factor corresponds to a leaf in $\T'$ which is labeled with a list
  of starting positions of its occurrences in $X$.

  We assume left-to-right orientation of the children of each node (e.g., lexicographic).
  A global occurrence list $\OL$ is stored being a concatenation of the lists of occurrences in
  all the leaves of the trie $\T'$ in pre-order.
  Each node $v$ stores, as $\OL(v)$, the occurrence list of leaves in its subtree represented as a pair of pointers
  to elements of the global list $\OL$.
  We enhance the occurrence list $\OL$ by a data structure for the following colored range listing problem.

  \begin{problem}{\textbf{Problem 2.} Colored range listing}
    Preprocess a sequence $A[1..N]$ of elements from $[1..S]$ so that, given a range $A[i..j]$,
    one can list all the distinct elements in that range.
  \end{problem}

  \begin{fact}[Muthukrishnan~\cite{Muthukrishnan2002}]
    A data structure for the colored range listing problem of $\Oh(N)$ size can be constructed
    in $\Oh(N+S)$ time and answers queries in $\Oh(k+1)$ time where $k$ is the number of distinct elements
    reported.
  \end{fact}
  For all nodes of $\T'$ we also compute the following values (for the purpose of this computation we replace each leaf $v$
  with $|\OL(v)|$ bogus leaves with single occurrences).

  \begin{fact}[Color set size, Hui~\cite{DBLP:conf/cpm/Hui92}]\label{fct:count}
    Given a rooted tree of size $N$ with $L$ leaves colored from $[1..S]$, in $\Oh(N+S)$ time 
    one can find for each node $u$ the number of distinct leaf colors in the subtree of $u$.
  \end{fact}
  We denote the resulting data structure as $\I$.

  \begin{theorem}
    The index $\I$ for a weighted sequence $X$ can be constructed in $\Oh(nz)$ time.
    It answers decision and counting variants of weighted pattern matching queries in $\Oh(m)$ time,
    and, if required, reports all occurrences of the pattern in $\Oh(m+|\Occ(P,X)|)$ time.
  \end{theorem}
  \begin{proof}
    The compacted trie $\T'$ can answer queries if $\Occ(P,X) \ne \emptyset$
    in $\Oh(m)$ time. We can use \cref{fct:count} to equip each explicit node with the number
    of positions where the string represented by the node occurs. This way, $|\Occ(P,X)|$
    can also be determined in $\Oh(m)$ time.
    With the aid of the data structure  for colored range listing, we can also report $\Occ(P,X)$ in time
    proportional to the number of reported elements.
  \mayqed\end{proof}

  \section{Applications of the Weighted Index}\label{sec:applications}
  In this section we present two non-trivial applications of the weighted index.
  In both cases we improve the time complexity of the previously known results by a factor of $z \log z$.

  \subsection{Weighted longest common prefixes and weighted prefix table}
  For a weighted sequence $X$ of length $n$ and a pair of indices $i,j$, $1 \le i,j \le n$, by $\wlcp(i,j)$
  we denote the length of the longest solid factor that occurs in $X$ at both positions $i$ and $j$.
  After some preprocessing our weighted index allows to answer such
  queries in $\Oh(z)$ time.

  \begin{theorem}\label{thm:wlcp}
    Given a weighted sequence $X$ of length $n$, after $\Oh(nz)$-time preprocessing we can answer $\wlcp(i,j)$
    queries for any $1 \le i,j \le n$ in $\Oh(z)$ time.
  \end{theorem}
  \begin{proof}
    For each position $i$ in $X$ we precompute the list of leaves $L(i)$ of the weighted index $\I$
    that contain $i$ in their occurrence lists.
    Prior to that, all leaves are numbered in pre-order, and the elements of $L(i)$ are stored in this order.
    By \cref{fct:maximal}, $|L(i)| \le z$ for each $i$.

    Observe that $\wlcp(i,j)$ is the maximum depth of a lowest common ancestor ($\lca$) of a leaf in $L(i)$ and a leaf in $L(j)$. 
    To determine this value, we merge the lists $L(i)$ and $L(j)$ according to the pre-order.
    The claim below (Lemma~4.6~in \cite{AlgorithmsOnStrings}) implies that, computing
    $\wlcp(i,j)$, it suffices to consider pairs of leaves that are adjacent in the resulting list.

    \begin{claim}
      If $l_1$, $l_2$ and $l_3$ are three leaves of a (compacted) trie such that $l_2$ follows $l_1$
      and $l_3$ follows $l_2$ in pre-order, then $\depth(\lca(l_1,l_3)) = \min(\depth(\lca(l_1,l_2)),\depth(\lca(l_2,l_3)))$.
    \end{claim}%
    Merging two sorted lists, each of length at most $z$, takes $\Oh(z)$ time.
    Finally let us recall that $\lca$-queries in a tree can be answered in $\Oh(1)$ time after linear-time
    preprocessing \cite{LCA,DBLP:journals/siamcomp/HarelT84}.
  \mayqed\end{proof}

  The weighted prefix table $\PREF[1..n]$ of $X$ is defined as $\PREF[i] = \wlcp(1,i)$; see \cite{DBLP:conf/cwords/BartonP15}.
  As a consequence of \cref{thm:wlcp} we obtain an $\Oh(nz)$-time algorithm for computing this table.
  It outperforms the algorithm of \cite{DBLP:conf/cwords/BartonP15}, which works in $\Oh(n z^2 \log z)$ time.

  \begin{theorem}
    The weighted prefix table $\PREF$ of a given weighted sequence of length $n$ can be computed in $\Oh(nz)$ time.
  \end{theorem}

  \subsection{Efficient computation of covers}\label{sec:covers}
  A \emph{cover} of a weighted sequence $X$ is a string $P$ whose occurrences as solid factors of $X$ cover all positions in $X$;
  see \cite{costas_weighted_suffix_tree_j}.
  More formally, if we define $\maxgap$ of an ordered set $A = \{a_1,\ldots,a_k\}$ (with $a_1 < \ldots < a_k$) as
  $$\maxgap(A) = \max\{a_i - a_{i-1}\,:\, i=2,\ldots,k\},$$
  then $P$ is a cover of $X$ if and only if
  $$1 \in \Occ(P,X) \quad\mbox{and}\quad \maxgap (\Occ(P,X) \cup \{n+1\}) \le |P|.$$
  Note that the former condition means exactly that $P$ is a solid prefix of $X$.
  An $\Oh(n)$-time algorithm computing a representation of all the covers of a weighted sequence
  under the assumption that $z = \Oh(1)$ was presented in \cite{costas_weighted_suffix_tree_j}.
  Here we show an algorithm that works in $\Oh(nz)$ time.

  The algorithm of \cite{costas_weighted_suffix_tree_j} uses a data structure (which we denote here as $\D$) to store a multiset
  of elements $A$ from the set $\{2,\ldots,n\}$ allowing three operations:
  \begin{enumerate}[(a)]
    \item initialisation with a given multiset of elements $A$;
    \item computing $\maxgap(\D) = \maxgap(A \cup \{1,n+1\})$ for the currently stored multiset $A$;
    \item removing a specified element from the currently stored multiset $A$.
  \end{enumerate}
  The data structure has $\Oh(n)$ size, executes operation (a) in $\Oh(|A|+n)$ time and supports operations
  (b) and (c) in constant time.
  It consists of: (1) an array $C[1..n+1]$ that counts the multiplicity of each element;
  (2) a list $L$ that stores all distinct elements of $A \cup \{1,n+1\}$ in ascending order and retains its maxgap; and
  (3) an array $P[1..n+1]$ that stores, for each distinct element of $A \cup \{1,n+1\}$, a pointer to its occurrence in $L$.

  The algorithm of \cite{costas_weighted_suffix_tree_j}, formulated in terms of our index $\I$,
  works as follows.
  For a node $v$ let $\D(v)$ be the $\D$-data structure storing the multiset $\OL(v) \setminus \{1\}$.
  The path from the root to each terminal node that represents a \emph{maximal solid prefix} of $X$ is traversed,
  and at each explicit node $v$ the data structure $\D(v)$ is computed.
  When we move from a node $v$ to its child $w$ on the path, from $\D(v)$ we remove all elements
  from $\OL(w')$ for $w'$ being children of $v$ other than $w$.
  Afterwards for the node $w$ we perform the following check, which we call $\covercheck(w)$:
  if $\maxgap(\D(w)) \le \depth(w)$, report the covers being prefixes of $\str(w)$ of length
  $[\max(\maxgap(\D(w)),\depth(v)+1)..\depth(w)]$.
  The whole procedure works in $\Oh(nz^2)$ time, as a single traversal works in linear time w.r.t.\ the
  size of the index and there are at most $z$ maximal solid prefixes of $X$ (\cref{fct:maximal}).

  Let us show how this algorithm can be implemented to run in $\Oh(nz)$ time.
  We will call an explicit node of $\I$ a \emph{prefix node} if it corresponds to a solid prefix of $X$.
  To implement the solution, it suffices for each prefix node to compute the $\D$-data structure
  and apply the $\covercheck$ routine.
  A prefix node will be called \emph{branching} if it has more than one child being a prefix node,
  and \emph{starting} if it is the root or its parent is branching.
  A maximal path going down the trie from a starting prefix node and passing only through non-starting prefix nodes
  will be called a \emph{covering} path.
  Considering the prefix node subtree of $\I$, which contains at most $z$ leaves and, consequently, at most $z-1$
  branching nodes, we make the following easy but important observation.

  \begin{observation}\label{obs:starting}
    There are $\Oh(z)$ covering paths and each prefix node belongs to exactly one of them.
  \end{observation}

  We will first show how to efficiently compute the $\D$-data structure
  for all starting prefix nodes and then argue that we can update the data structure efficiently along each covering path.

  \begin{lemma}\label{lem:starting}
    $\D(v)$ for all starting prefix nodes $v$ can be computed in $\Oh(nz)$ time.
  \end{lemma}
  \begin{proof}
    We compute the $C$-arrays of the data structures for all the prefix nodes $v$ from the bottom to the top,
    storing them for starting nodes only.
    In the end we shall construct the $L$-lists and $P$-arrays for the starting prefix nodes, which takes $\Oh(nz)$ time
    as there are at most $z$ starting prefix nodes (\cref{obs:starting}).
    
    For a leaf $v$ we initialise its $C$-array with $\OL(v)$ in $\Oh(|\OL(v)|+n)$ time.
    
    For a non-branching prefix node $v$ with prefix node child $w$, we take the $C$-array of the latter
    and insert all elements of $\OL(w')$ for the remaining children $w'$ of $v$.
    The time complexity is proportional to the sum of lengths of those lists $\OL(w')$.
    Note that we do not need to store the $C$-array of $w$ any more.
    
    Finally, for a branching prefix node $v$ we add up the counters of the $C$-arrays of all
    its prefix node children and insert all elements of $\OL(w')$ for the remaining children $w'$ of $v$.
    This time the complexity is proportional to the sum of lengths of the lists $\OL(w')$ and $n$ times
    the number of prefix node children.

    Let us analyse the time complexity of the algorithm.
    Note that each element of the global list $\OL$ will be added to some $C$-array at most once.
    Also when two or more $C$-arrays are added component-wise at a branching prefix node, all of them
    correspond to starting prefix nodes and each starting prefix node participates in such operation at most once.
    Moreover, note that we maintain exactly one array for each covering path.
    This way, the total time and space complexity is $\Oh(nz)$.
  \end{proof}

  The main part of the solution is actually quite similar to the computations of \cref{lem:starting}.
  This time, however, we move top-down along the tree, which allows us to compute maxgaps in the
  $\D$-data structures.

  \begin{lemma}\label{lem:maxgaps}
    The values $\maxgap(\D(v))$ for all prefix nodes can be computed in $\Oh(nz)$ time.
  \end{lemma}
  \begin{proof}
    We first perform the computations of \cref{lem:starting}.
    Let us fix a starting prefix node. 
    We compute its maxgap and go along the covering path that it belongs to, maintaining the $\D$-data structure.
    Every node on the path (except the last one) has exactly one prefix node child.
    Hence, when going from a prefix node $v$ to its child, prefix node $w$, we simply
    remove from $\D(v)$ the elements of $\OL(w')$ for all the remaining children $w'$ of $v$.
    This allows us to compute $\maxgap(\D(w))$.

    In total each element of the global occurrence list $\OL$ will be removed from a $\D$-data structure
    at most once, which yields time complexity $\Oh(nz)$.
  \end{proof}

  \begin{theorem}
    A representation of size $\Oh(nz)$ of all covers of a weighted sequence $X$ of length $n$
    can be computed in $\Oh(nz)$ time.
    In particular, all shortest covers of $X$ can be determined in $\Oh(nz)$ time.
  \end{theorem}
  \begin{proof}
    To annotate all the covers on the edges of the index, we compute the maxgaps for
    all the prefix nodes using \cref{lem:maxgaps} and then apply the constant-time $\covercheck$ routine
    for each of the nodes.
    As for the shortest covers, there are at most $z$ of them (as there are at most $z$ different solid prefixes
    of $X$ of a specified length, each with probability of occurrence at least \fr), so they can all be listed explicitly in $\Oh(nz)$ time and space.
  \mayqed\end{proof}

  \section{Conclusions}
  We have presented an index for weighted pattern matching queries which for a constant-sized alphabet
  has $\Oh(nz)$ size and admits $\Oh(nz)$ construction time.
  It answers queries in optimal $\Oh(m+\mathit{Occ})$ time.
  We have also mentioned two applications of the weighted index.
  Our index outperforms the previously existing solutions by a factor of $z \log z$ in the complexity.

  Let us briefly discuss how to adapt our index to a general integer alphabet.
  The size of the input is then the total length $R$ of the lists in the representation of the weighted sequence.
  In the construction of the solid factor trie we need the list at each position to be ordered according to the probabilities of letters.
  As the size of each list to be sorted is $\min(z,\sigma)$ (at most $z$ letters can have probability at least \fr),
  the sorting requires $\Oh(R \log \min(\sigma,z))$ time.
  The construction of a suffix tree of a tree of \cite{journals/ieice/Shibuya03} works for any integer alphabet.
  Finally, our weighted index is a compacted trie with children of a node being indexed
  by the letter of the alphabet.
  Hence, to avoid an increase of the complexity of a query for a particular child of a node, for a general alphabet
  one requires to store a hash table of children.
  With perfect hashing \cite{DBLP:journals/jacm/FredmanKS84} the complexity does not increase but becomes
  randomized (Las Vegas, running time w.h.p.).

  An open question is whether our weighted index, constructed for a given $z$, can be adapted to answer
  weighted pattern matching queries for $z^{\prime}<z$, as it is the case of \cite{DBLP:journals/corr/ThankachanPSB15}.

\bibliographystyle{plain}
  \bibliography{weighted_index}

\appendix
  \newpage
  \section{Illustrations for sample weighted sequence}\label{app:figures}
  For succinctness of presentation, for a non-solid position $i$ (i.e., a position with at least two different letters), the set of pairs  
    is denoted by $[(s_{j_1}, \pi_i(s_{j_1})),\ldots,(s_{j_k}, \pi_i(s_{j_k}))]$; for a solid position $i$ it is
    simply denoted by the letter $s_j$ with $\pi_i(s_j)=1$. 
  
  \begin{example}\label{ex:1}
  Example of $\T$ for $z=4$ and the weighted sequence
  $$X=[(\texttt{a},0.5),(\texttt{b},0.5)]\texttt{bab}[(\texttt{a},0.5),(\texttt{b},0.5)][(\texttt{a},0.5),(\texttt{b},0.5)]\texttt{aaba}.$$
  Among a few heavy strings of $X$, we select $\X=\texttt{ababaaaaba}$.
  
  \bigskip
  \begin{center}
  \includegraphics{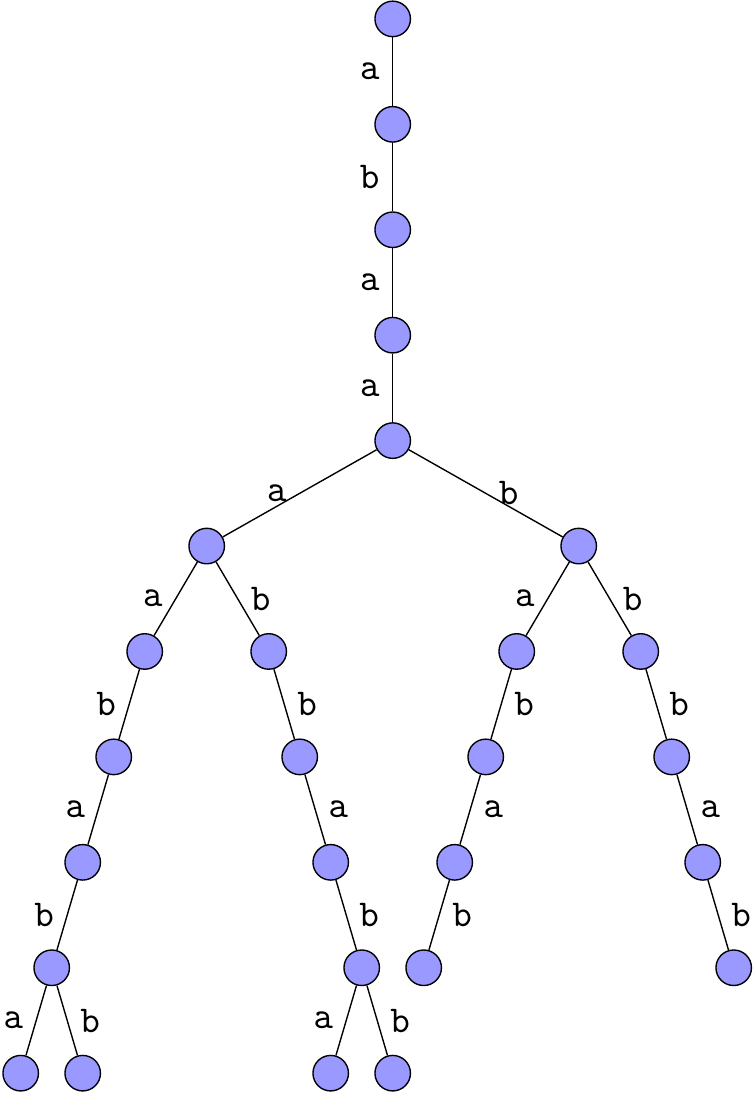}
  \end{center}
  \end{example}

  \newpage
  
  \begin{example}\label{ex:2}
  Example of $S(\T)$ (labels of edges to terminals are appended with a $ \$$) for $\T$ shown in \cref{ex:1}.
  
  \bigskip
  \begin{center}
    \includegraphics[angle=270,origin=c]{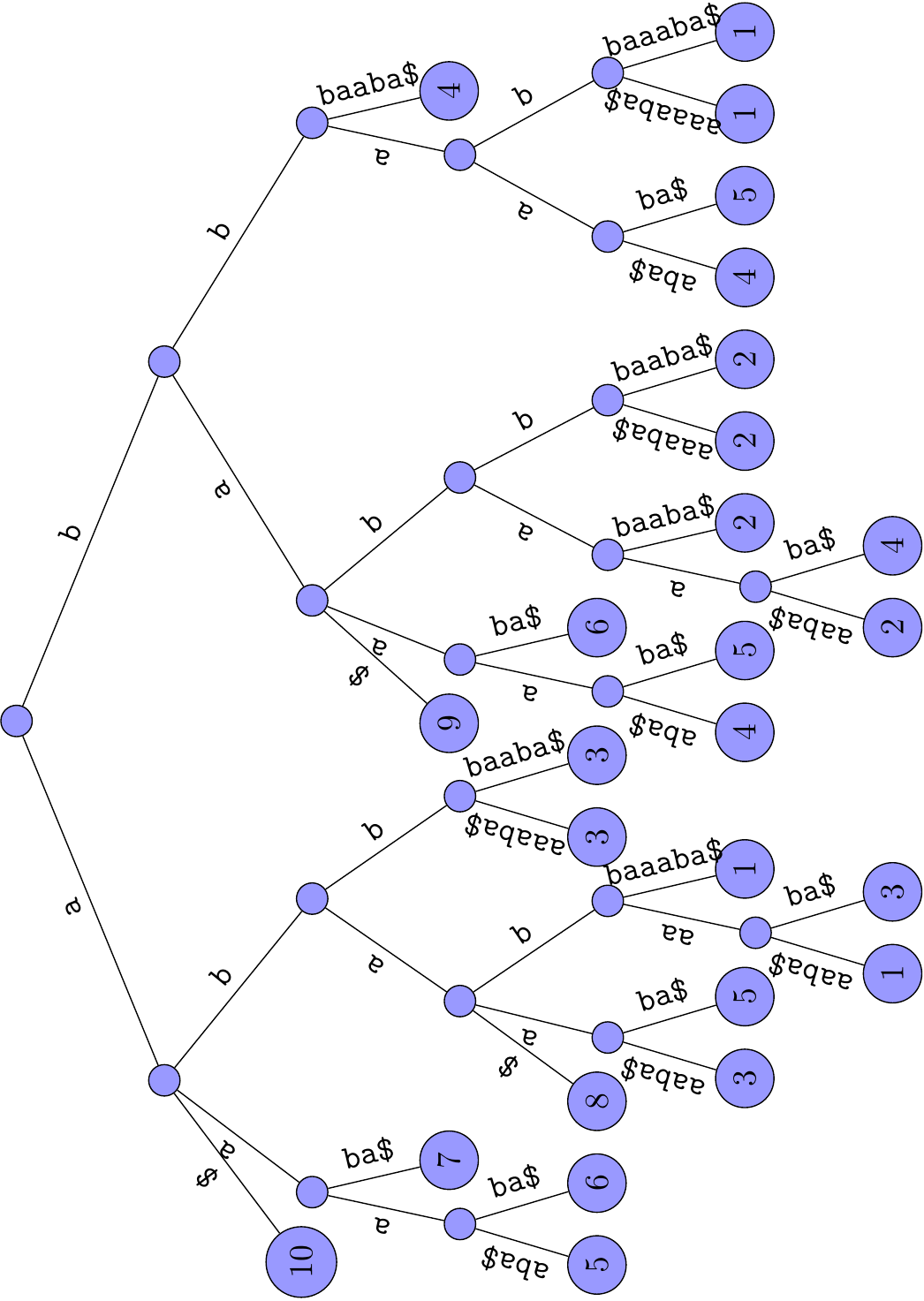}
  \end{center}
  \end{example}

  \newpage

\begin{example}
  Example of $\I$ for $\S(\T)$ shown in \cref{ex:2}. Note that compared to $\S(\T)$ all leaves with label 1 are ``lifted'' by 5 characters up.
  Consequently, a node with label $\texttt{ababa}$ becomes explicit, and a node with label $\texttt{ababaa}$ becomes implicit.
  
  \bigskip
  \begin{center}
    \includegraphics[angle=270,origin=c]{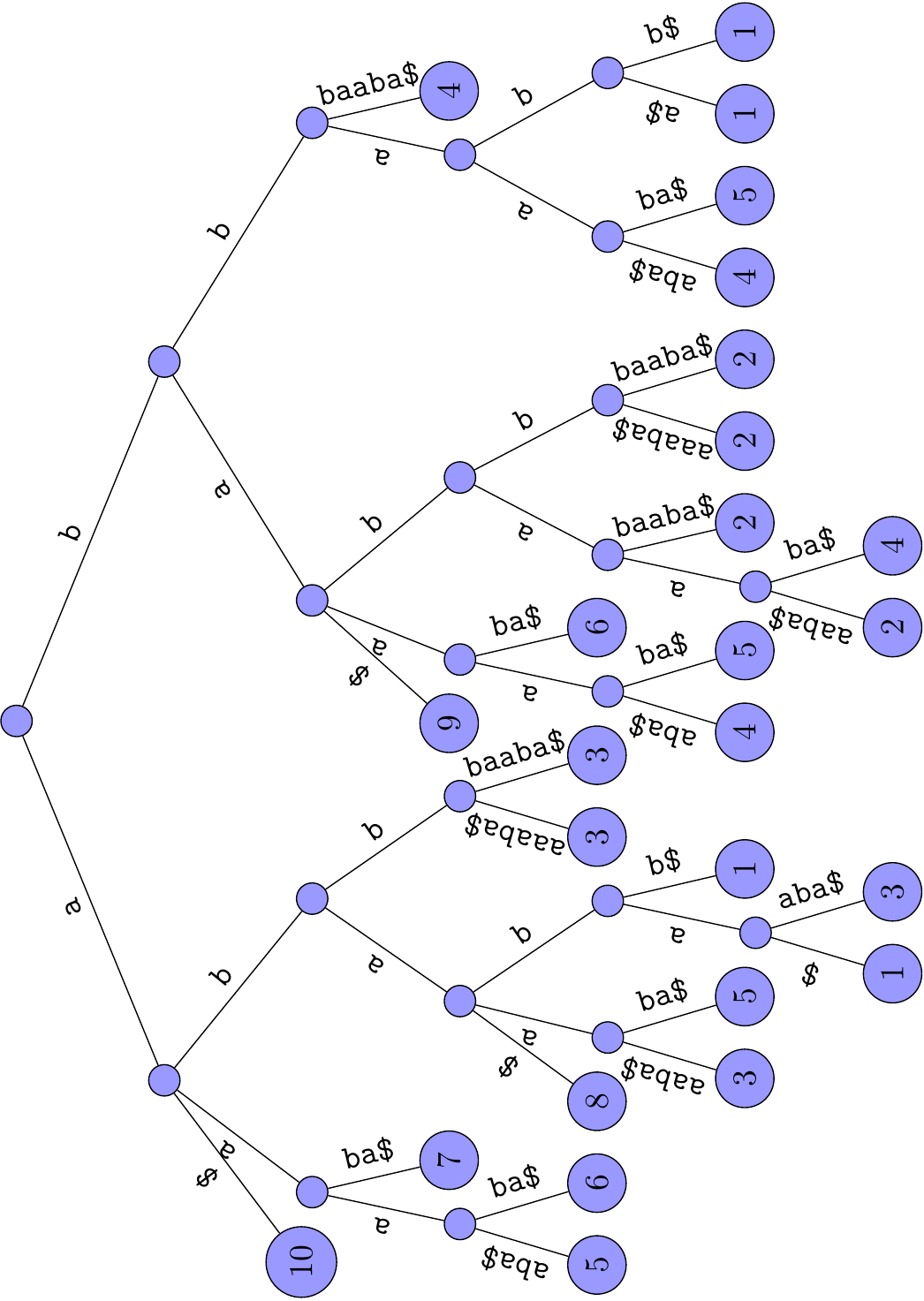}
  \end{center}
  \end{example}

  \bigskip
  \begin{example}
    For the weighted sequence from \cref{ex:1} and $z=4$, the weighted prefix table is:
    $$[5,\,1,\,5,\,3,\,3,\,1,\,1,\,3,\,1,\,1].$$
    The weighted sequence has exactly one cover, \texttt{aba}.
  \end{example}
    
\end{document}